\documentclass[11pt, oneside]{article}   	
\usepackage{amsmath, amsthm, amsfonts, amssymb}
\usepackage[colorlinks,citecolor=blue,urlcolor=blue]{hyperref}
\newcommand\Ex{{\mathbb E}}

\newcommand\F{{\mathcal F}}

\newcommand\Prob{{\mathbb P}}
\newcommand\chitilde{\tilde{\chi}}
\newcommand\Normal{{\mathcal N}}

\newcommand\cB{{\mathcal B}}

\newcommand\E{{\mathcal E}}

\newcommand\cJ{{\mathcal J}}

\newcommand\cP{{\mathcal P}}

\newcommand\R{{\mathbb R}}

\newcommand\C{{\mathbb C}}

\newcommand\dto{\overset{d}{\to }}
\newcommand\wto{\overset{w}{\to }}

\newcommand\one{{\bf 1}}

\newcommand\bra[1]{\langle #1 \rangle}

\DeclareMathOperator{\DP}{DP}

\DeclareMathOperator{\Beta}{Beta}
\DeclareMathOperator{\Dir}{Dir}

\DeclareMathOperator{\tr}{tr}

\newtheorem{theorem}{Theorem}[section]

\newtheorem{lemma}[theorem]{Lemma}
\newtheorem{proposition}[theorem]{Proposition}

\theoremstyle{definition}
\newtheorem{definition}[theorem]{Definition}

\theoremstyle{remark}
\newtheorem{remark}[theorem]{Remark}
\newtheorem{example}[theorem]{Example}

\title{The spectral measures of random Jacobi matrices related to beta ensembles at high temperature and Dirichlet processes}

\author{
Fumihiko Nakano\footnote{Mathematical Institute, Tohoku University, Sendai,  Japan.
\newline Email: fumihiko.nakano.e4@tohoku.ac.jp}
\and
Hoang Dung Trinh\footnote{Faculty of Mathematics Mechanics Informatics, University of Science, Vietnam National University, Hanoi, Vietnam.
\newline Email: thdung.hus@gmail.com} 
\and
Khanh Duy Trinh\footnote{Global Center for Science and Engineering, Waseda University, Japan.
\newline
Email: trinh@aoni.waseda.jp 
} 
}

\begin{document}

\maketitle
\begin{abstract}
In a high temperature regime where $\beta N \to 2c$, the empirical distribution of the eigenvalues of Gaussian beta ensembles, beta Laguerre ensembles and beta Jacobi ensembles converges to a limiting measure which is related to associated Hermite polynomials, associated Laguerre polynomials and associated Jacobi polynomials, respectively. Here $\beta$ is the inverse temperature parameter, $N$ is the system size and $c>0$ is a given constant. This paper studies the spectral measure of the random tridiagonal matrix model of the three classical beta ensembles. We show that  in the high temperature regime, the spectral measure converges in distribution to a Dirichlet process with base distribution being the limiting distribution, and scaling parameter $c$. Consequently, the spectral measure of a related semi-infinite Jacobi matrix coincides with that Dirichlet process, which provides examples of random Jacobi matrices with explicit spectral measures. 

		
\medskip
	
%
%
%
%
\end{abstract}

\section{Introduction}

Beta ensembles in the so-called high temperature regime, the regime where $N \to \infty$ with $\beta N \to 2c \in (0, \infty)$, have gained much attention recently. Here $\beta > 0$ is the inverse temperature parameter, $N$ denotes the system size and $c \in (0, \infty)$ is a given constant. In this regime, the empirical distribution of the eigenvalues converges weakly to a limiting probability measure $\rho_c$ depending on the parameter $c$ and the weight or the potential of the model \cite{Lambert-2021, Nakano-Trinh-2020}. In the three classical cases (Gaussian, Laguerre and Jacobi), the measure $\rho_c$ is related to associated Hermite polynomials, associated Laguerre polynomials and associated Jacobi polynomials, respectively \cite{Allez12, Trinh-Trinh-Jacobi, Trinh-Trinh-2021}.

The paper studies the spectral measure of the random Jacobi matrix (symmetric tridiagonal matrices) model of the three classical beta ensembles. We show that in the high temperature regime, the spectral measure converges in distribution to a Dirichlet process with base distribution $\rho_c$ and scaling parameter $c$. Here a Dirichlet process with base distribution $\rho$ and scaling parameter $c>0$, denoted by $\DP(\rho, c)$, is a random probability measure $P$ with the property that for any measurable partition $\R = A_1 \cup A_2 \cup \cdots \cup A_n$, the vector 
\[
	(P(A_1), \dots, P(A_n)) \sim \Dir(c\rho(A_1), \dots, c\rho(A_n)),
\]
has the Dirichlet distribution with parameters $(c\rho(A_1), \dots, c\rho(A_n))$.

Let us get into more detail by introducing the result in the Gaussian case. The random tridiagonal matrix model introduced in \cite{DE02} has the following form 
\begin{align*}
	\tilde H_{N, \beta} &=\frac{\sqrt 2}{\sqrt{\beta N}} \begin{pmatrix}
		a_1^{(N)}	&b_1^{(N)}	\\
		b_1^{(N)}	&a_2^{(N)}		&b_2^{(N)}\\
		&	\ddots	&\ddots	&\ddots\\
			&&b_{N-1}^{(N)}	&a_N^{(N)}
	\end{pmatrix}\\
	&
	\sim
	\frac{\sqrt 2}{\sqrt{\beta N}}\begin{pmatrix}
		\Normal(0,1)	&\chitilde_{(N-1)\beta}	\\
		\chitilde_{(N-1)\beta}	&\Normal(0,1)		&\chitilde_{(N-2)\beta}\\
		&	\ddots	&\ddots	&\ddots\\
			&&\chitilde_{\beta}	&\Normal(0,1)
	\end{pmatrix}.
\end{align*}
To be more precise, the random variables $\{a_1^{(N)}, \dots, a_N^{(N)}, b_1^{(N)}, \dots, b_{N-1}^{(N)}\}$ are independent with $a_i^{(N)} \sim \Normal(0,1)$ having the standard normal distribution, and $b_j^{(N)}\sim \chitilde_{(N-j)\beta}$ having the $1/\sqrt 2$ chi distribution with $(N-j)\beta$ degrees of freedom.  The eigenvalues $\lambda_1, \dots, \lambda_N$ of $\tilde H_{N, \beta}$ follow the Gaussian beta ensembles, that is, their joint density is proportional to 
\[
	\prod_{i<j} |\lambda_j - \lambda_i|^\beta \prod_{l=1}^N e^{-\frac {\beta N}{4} \lambda_l^2}.
\]

Spectral properties of Gaussian beta ensembles have been studied by reading off the random tridiagonal matrix model \cite{DE06, Ramirez-Rider-Virag-2011, Valko-Virag-2009}.
For fixed $\beta > 0$, the empirical distribution 
\[
	L_N = \frac{1}{N} \sum_{i=1}^N \delta_{\lambda_i}
\]
converges weakly to the standard semi-circle distribution, almost surely \cite{DE06}, which is called Wigner's semi-circle law. Here $\delta_\lambda$ is the Dirac measure. Even in case the parameter $\beta$ depends on $N$, Wigner's semi-cirlce law holds as long as $\beta N \to \infty$ \cite{Trinh-2019}.

In a high temperature regime where $\beta N \to 2c$, with given constant $c  \in (0, \infty)$, the empirical distribution $L_N$ converges to a limiting probability measure which is related to associated Hermite polynomials \cite{Allez12}. In this regime, for simplicity, let $\beta = 2c/N$ and consider the scaled version 
\begin{align*}
	H_{N} = \begin{pmatrix}
		a_1^{(N)}	&b_1^{(N)}	\\
		b_1^{(N)}	&a_2^{(N)}		&b_2^{(N)}\\
		&	\ddots	&\ddots	&\ddots\\
			&&b_{N-1}^{(N)}	&a_N^{(N)}
	\end{pmatrix}
	\sim
\begin{pmatrix}
		\Normal(0,1)	&\chitilde_{(N-1)\beta}	\\
		\chitilde_{(N-1)\beta}	&\Normal(0,1)		&\chitilde_{(N-2)\beta}\\
		&	\ddots	&\ddots	&\ddots\\
			&&\chitilde_{\beta}	&\Normal(0,1)
	\end{pmatrix}
\end{align*}
whose eigenvalues have the joint density proportional to 
\begin{equation}\label{GbE-hightemp}
	\prod_{i<j} |\lambda_j - \lambda_i|^{\frac{2c}N} \prod_{l=1}^N e^{-\frac {1}{2} \lambda_l^2}.
\end{equation}
Then the empirical distribution of those eigenvalues converges weakly to a probability measure $\rho_c$ with density given by 
\begin{equation}\label{rho-c-Gauss}
		\rho_c(x) = \frac{e^{-x^2/2}}{\sqrt{2\pi}}\frac{1}{|\hat f_c(x)|^2},\quad
		\text{where }
		\hat f_c(x) =  \sqrt{\frac{c}{\Gamma(c)}} \int_0^\infty t^{c - 1} 
		e^{-\frac {t^2}{2} + \sqrt{-1}xt} dt.		
\end{equation}
There are several approaches (see \cite{Allez12, DS15, Forrester-2021, NTT-2025}, for example) to identify the measure $\rho_c$ which is the probability measure of associated Hermite polynomials \cite{Askey-Wimp-1984}.

Next, let us define the spectral measure of a symmetric tridiagonal matrix, called a Jacobi matrix. To a semi-infinite Jacobi matrix $J$,
\[
	J = \begin{pmatrix}
		a_1	&b_1	\\
		b_1	&a_2		&b_2\\
		&	\ddots	&\ddots	&\ddots
	\end{pmatrix}, \quad a_i \in \R, b_j > 0,
\]
there is a probability measure $\mu$ on the real line  satisfying the following moment condition 
\begin{equation}\label{spectral-measure-moment}
	\int_\R x^n d\mu(x) =  J^n(1, 1), \quad n = 0, 1, 2, \dots.
\end{equation}
That measure orthogonalizes the sequence of polynomials 
\begin{align*}
	&P_0(x) = 1; P_1(x) = x - a_1;\\
	&P_{n+1}(x) = x P_n(x) - a_{n+1} P_{n}(x) - b_n^2 P_{n-1}(x), \quad (n \ge 1).
\end{align*}
When the measure $\mu$ is unique, or equivalently is determined by moments, we call it the spectral measure of $J$. A useful sufficient condition for the unicity is given by 
\begin{equation}\label{sp-sufficient}
\sum_{n=1}^\infty \frac 1 {b_n} = \infty,
\end{equation}
\cite[Corollary 3.8.9]{Simon-book-2011}.

We give two examples here. First, the standard semi-circle distribution is the spectral measure of 
\[
	J_{free} = \begin{pmatrix}
		0	&1	\\
		1	&0		&1\\
		&	\ddots	&\ddots	&\ddots
	\end{pmatrix}, \quad (a_i = 0, b_j=1),
\] 
and is related to the Chebyshev polynomials of the second kind. Second, the limiting probability measure $\rho_c$ in~\eqref{rho-c-Gauss} is the spectral measure of 
\[
	\begin{pmatrix}
		0	&\sqrt{c+1}	\\
		\sqrt{c+1}	&0		&\sqrt{c+2}\\
		&	\ddots	&\ddots	&\ddots
	\end{pmatrix}, \quad (a_i = 0, b_j=\sqrt{c+j}),
\] 
which orthogonalizes associated Hermite polynomials $\{P_n^{(c)}\}_{n \ge 0}$,
\begin{align*}
	&P_0^{(c)}(x) = 1, \quad P_1^{(c)}(x) = x,\\
	&P^{(c)}_{n+1}(x) = x P^{(c)}_n(x)  - (n+c) P^{(c)}_{n-1}(x), \quad (n \ge 1).
\end{align*}

In case of finite Jacobi matrices, the spectral measure can be expressed in terms of their eigenvalues and eigenvectors.
Take the random Jacobi matrix $\tilde H_{N, \beta}$, for example. Its spectral measure, denoted by $\widetilde {sp}_N$, is supported on the eigenvalues $\lambda_1, \dots, \lambda_N$ with weights $w_i = v_i(1)^2$, where $v_1, \dots, v_N$ are the corresponding normalized eigenvectors,
\[
	\widetilde {sp}_N = \sum_{i=1}^N w_i \delta_{\lambda_i}.
\]
Moreover, the weights $\{w_i\}_{i=1}^N$ are independent of the eigenvalues and have the symmetric Dirichlet distribution with parameter $\beta / 2$. As $N \to \infty$ with
$\beta N \to \infty$, it is clear that 
\[
	\frac{\sqrt 2}{\sqrt {N \beta}}a_i^{(N)} \sim \frac{\sqrt 2}{\sqrt {N \beta}} \Normal(0,1) \to 0, \quad\frac{\sqrt 2}{\sqrt {N \beta}} b_i^{(N)} \sim \frac{\sqrt 2}{\sqrt {N \beta}} \chitilde_{(N-i)\beta} \to 1.
\]
Here the convergences hold in probability for any fixed $i$. Roughly speaking, the random matrix $\tilde H_{N, \beta}$ converges in probability to $J_{free}$, from which Wigner's semi-circle law for the spectral measure of $\tilde H_{N, \beta}$ holds \cite{Trinh-ojm2018}. Note that as $N \to \infty$ with
$\beta N \to \infty$, both the empirical distribution and the spectral measure converge to the same limit.

In the high temperature regime $\beta =  2c/N$, the random matrix $H_N$ converges in distribution to the following semi-infinite Jacobi matrix 
\begin{equation}\label{Hc}
	H_c = 	\begin{pmatrix}
		a_1^{(\infty)}	&b_1^{(\infty)}	\\
		b_1^{(\infty)}	&a_2^{(\infty)}		&b_2^{(\infty)}\\
		&	\ddots	&\ddots	&\ddots
	\end{pmatrix}\sim\begin{pmatrix}
		\Normal(0,1)	&\chitilde_{2c}	\\
		\chitilde_{2c}	&\Normal(0,1)		&\chitilde_{2c}\\
		&	\ddots	&\ddots	&\ddots
	\end{pmatrix}.
\end{equation}
Here $\{a_n^{(\infty)}\}_{n=1}^\infty$ and $\{b_n^{(\infty)}\}_{n=1}^\infty$ are two independent sequences of i.i.d.\ (independent identically distributed) random variables. By independence, for any $k$, 
\[
	\{a_1^{(N)}, \dots, a_k^{(N)}, b_1^{(N)}, \dots, b_k^{(N)}\} \dto \{a_1^{(\infty)}, \dots, a_k^{(\infty)}, b_1^{(\infty)}, \dots, b_k^{(\infty)}\}
\]
as the joint convergence in distribution of random variables.
 Let $\nu_c$ be the spectral measure of $H_c$, which is unique almost surely by taking into account the sufficient condition~\eqref{sp-sufficient}. From the joint convergence in distribution of the entries, we deduce that the spectral measure $sp_N$ of $H_N$ converges in distribution to the random probability measure $\nu_c$ (see Lemma~\ref{lem:nuc}).

The aim of this paper is to identify the limit of the spectral measure $sp_N$ of $H_N$  by using the explicit expression 
\begin{equation}\label{spN}
	sp_N = \sum_{i=1}^N w_i \delta_{\lambda_i},
\end{equation}
where $\lambda_1, \dots, \lambda_N$ are distributed as the Gaussian beta ensemble~\eqref{GbE-hightemp}, independent of the vector of weights $(w_1, \dots, w_N)$ which has the symmetric Dirichlet distribution with parameter $c/N$. We show that the sequence $\{sp_N\}$ converges in distribution to a Dirichlet process with base distribution $\rho_c$ and scaling parameter $c$.  A more general result is stated as follows.

\begin{theorem}\label{thm:intro-general}
Let $c > 0$ be given.
	For each $N$, let $(w_1, \dots, w_N)$ have the symmetric Dirichlet distribution with parameter $c/N$ and be independent of  random variables $\lambda_1, \dots, \lambda_N$. Assume that the empirical distribution $L_N = N^{-1} \sum_{i=1}^N \delta_{\lambda_i}$ converges weakly to a limiting probability measure $\rho$ as $N \to \infty$, in probability. Then the random probability measure (the spectral measure) 
\[
	sp_N := \sum_{i=1}^N w_i \delta_{\lambda_i}
\]
converges in distribution to a Dirichlet process with base distribution $\rho$ and scaling parameter $c$.
\end{theorem}

Consequently, the spectral measure $\nu_c$ of the semi-infinite Jacobi matrix $H_c$ is a Dirichlet process $\DP(\rho_c, c)$. Moreover, from an explicit construction of the Dirichlet process in \cite{Ferguson-1973, Sethuraman-1994}, $\nu_c$ has the same distribution with a random discrete probability measure 
\[
	P = \sum_{i=1}^\infty Q_i \delta_{Y_i},
\]
where $Y = (Y_1, Y_2, \dots)$ is a sequence of i.i.d.\ random variables with the common distribution $\rho_c$, which is independent of the random weights $Q = (Q_1, Q_2, \dots)$. Here $Q$ has the Poisson-Dirichlet distribution with parameter $c$  in the infinite-dimensional simplex 
\[
	\Sigma = \left\{x = (x_1, x_2, \dots) : x_1 \ge x_2 \ge \cdots \ge 0, \sum x_i = 1\right\}.
\]

We summarize the result on the spectral measure of the semi-infinite Jacobi matrix $H_c$ in the following theorem.
\begin{theorem}\label{thm:main-intro} 
	The spectral measure $\nu_c$ of the random Jacobi matrix $H_c$~\eqref{Hc} is the Dirichlet process $\DP(\rho_c, c)$. Consequently, the spectral measure $\nu_c$ is discrete, almost surely, and is supported on an i.i.d.\ sequence from the common distribution $\rho_c$.
\end{theorem}

The above theorem provides a specific example of a family of random Jacobi matrices whose spectral measures are explicitly given. 
Analogous results in the Laguerre case and the Jacobi case will be stated in Section~\ref{sect:classical}.

Let us sketch ideas in the proof of Theorem~\ref{thm:intro-general}. To begin with, note that the convergence in distribution of random probability measures in Theorem~\ref{thm:intro-general} means that for any bounded continuous function $f \colon \R \to \R$, 
\begin{equation}\label{continuous-test}
	\bra{sp_N, f} = \sum_{i=1}^N w_i f(\lambda_i) \dto \bra{P, f},
\end{equation}
where $P\sim \DP(\rho, c)$ (see Section~\ref{sect:rpm}). Here $\bra{\mu, f}$ denotes the integral $\int f d\mu$ of an integrable function $f$ with respect to the measure $\mu$. Next, the distribution of $\bra{P, f}$ is the Markov--Krein transform of the distribution of $f$ under $\rho$, or equivalently,
\begin{equation}
	\Ex[(z - \bra{P, f})^{-c}] = \exp\left(	- c \int_\R \log(z - f(u)) d\rho(u)\right), \quad z \in \C \setminus \R.
\end{equation}
Then to prove the convergence \eqref{continuous-test}, the idea is to show that any  limit point of the sequence $\{\bra{sp_N, f}\}$ satisfies the above Markov--Krein relation. Details will be given in Section~\ref{sect:proof} after introducing Dirichlet processes and the Markov--Krein relation in Section~\ref{sect:MKT}.

\section{Diriclet processes and the Markov--Krein transform}
\label{sect:MKT}
Let $\cP(\R)$ be the space of all probability measures on $(\R, \cB(\R))$, where $\cB(\R)$ denotes the Borel $\sigma$-field of $\R$. A sequence of probability measures $\{\mu_N \}_{N = 1}^\infty$ is said to converge weakly to $\mu \in \cP(\R)$, denoted by $\mu_N \wto \mu$, if for all bounded continuous functions $f \colon \R \to \R$, 
\[
	\bra{\mu_N, f} \to \bra{\mu, f}\quad \text{as} \quad N\to \infty.
\]
The topology of weak convergence on $\cP(\R)$ can be metrizable by the L\'evy--Prokhorov metric  making it  a complete separable metric space. A random probability measure $\xi$ is defined to be a measurable map from some probability space $(\Omega, \F, \Prob)$ to $(\cP(\R), \cB(\cP(\R))$. Here $\cB(\cP(\R))$ is the Borel $\sigma$-field on $\cP(\R)$. 

Let $\xi$ be a random probability measure. Then for any Borel set $B \in \cB(\R)$, $\xi(B)$ is a real random variable. For any bounded measurable function $f$, the integral $\bra{\xi, f}$ is also a random variable.  As random elements on a separable metric space, concepts of almost sure convergence, convergence in probability and convergence in distribution of random probability measures are defined as usual. Equivalent conditions will be given in Definition~\ref{def:rm}.

Let $\alpha$ be a probability measure on $\R$ and $c > 0$ be a positive number.
A Dirichlet process $P$ with base distribution $\alpha$ and scaling parameter $c$, denoted by $P \sim \DP(\alpha, c)$, is a random probability measure such that for any finite measurable partition $\R = A_1 \cup A_2 \cup \cdots \cup A_n$, the vector $(P(A_1), \dots, P(A_n))$ has the Dirichlet distribution $\Dir(c\alpha(A_1), \dots, c \alpha(A_n))$.

Recall from the introduction that an explicit construction of the Dirichlet process is given by the following formula \cite{Ferguson-1973, Sethuraman-1994}
\[
	P = \sum_{i=1}^\infty Q_i \delta_{Y_i},
\]
where $Y = (Y_1, Y_2, \dots)$ is a sequence of i.i.d.\ random variables with the common distribution $\alpha$, and $Q = (Q_1, Q_2, \dots)$ is a random point of the infinite-dimensional simplex 
\[
	\Sigma = \left\{x = (x_1, x_2, \dots) : x_1 \ge x_2 \ge \cdots \ge 0, \sum x_i = 1\right\},
\]
that is independent of $Y$ and has the Poisson-Dirichlet distribution with parameter $c$. The random point $Q$ can be viewed as the limit of a sequence of the ordered symmetric Dirichlet distribution \cite{Kingman-1975}. Indeed, let 
$w^{(N)} = (w_1^{(N)}, \dots, w_N^{(N)})$ have the symmetric Dirichlet distribution with parameter $c/N$. Arrange $w^{(N)}$ in descending order to get the vector 
\[
	p_1^{(N)} \ge p_2^{(N)} \ge \cdots \ge p_N^{(N)}.
\]
Then for any finite $k$, as $N \to \infty$, 
\[
	(p_1^{(N)}, \dots, p_k^{(N)}) \dto (Q_1, \dots, Q_k).
\]

Let $f \colon \R \to \R$ be a measurable function  satisfying
\begin{equation}\label{integrability}
	\int_\R \log(1 + |f(x)|) d\alpha(x) < \infty.
\end{equation}
That condition implies the finiteness of the integral $\bra{ P, |f|}$, almost surely \cite{Feigin-Tweedie-1989}. Then the following relation holds \cite{Cifarelli-Regazzini-1990}
\begin{equation}\label{MKR-original}
	\Ex[(z - \bra{P, f})^{-c}] = \exp\left(	- c \int_\R \log(z - f(u)) d\alpha(u)\right), \quad z \in \C \setminus \R.
\end{equation}
Let $\nu$ be the distribution of $\bra{ P, f}$, and $\mu$ be the distribution of $f$ under the measure $\alpha$. The above relation tells us a relation between two probability measures $\nu$ and $\mu$, 
\begin{equation}\label{MKR}
	\int_\R \frac{1}{(z - t)^c} d\nu(t) = \exp\left(	- c \int_\R \log(z - u) d\mu(u)\right), \quad z \in \C \setminus \R,
\end{equation}
which uniquely determines each other \cite{LR2004}.
We say that  $\nu$ and $\mu$ are linked by the Markov--Krein relation (MKR) with parameter $c$, and call $\nu$ the Markov--Krein transform (MKT) of $\mu$.

We are going to give a simple proof of that MKR. Let us begin with an example on the Dirichlet distribution. 	
\begin{example}\label{ex:Dirichlet-distribution}
Let $w = (w_1, \dots, w_N)$ have the Dirichlet distribution with parameters $(\tau_1, \dots, \tau_N)$, where $\tau_i > 0, i = 1, \dots, N$, and $a = (a_1, \dots, a_N) \in \R^N$ be a non-random vector. Let $\nu$ be the distribution of the random variable $(w_1 a_1 + \cdots + w_N a_N)$, and $\mu = \frac{1}{c}\sum_{i=1}^N \tau_i \delta_{a_i}$ be a discrete probability measure on $\R$, where $c = \tau_1 + \cdots + \tau_N$. Then $\nu$ and $\mu$ are linked by the Markov--Krein relation with parameter $c$, that is, for $z \in \C \setminus \R$,
\[
	\int_\R \frac{1}{(z - t)^{c}} d\nu(t) = \Ex[(z - (w_1 a_1 + \cdots + w_N a_N))^{-c}] =  \prod_{i=1}^N \left(\frac1{z - a_i}\right)^{\tau_i}.
\]
A fundamental proof of that identity can be found in \cite{Fourati-2011}.
\end{example}

\begin{proof}[Proof of the MKR~\eqref{MKR-original}]

When $f$ is a simple function, the relation follows from Example~\ref{ex:Dirichlet-distribution}. Indeed, assume that
\[
	f = \sum_{i=1}^N a_i \one_{A_i},
\]
where $\{a_i\}_{i=1}^N$ are real numbers and $\{A_i\}_{i=1}^N$ are a partition of $\R$. Here $\one_A$ denotes the indicator function of a set $A$. It is clear that
\[
	\bra{P, f} = \sum_{i=1}^N a_i P(A_i).
\]
By the definition of Dirichlet processes, the vector $(P(A_1), \dots, P(A_N))$ has the Dirichlet distribution with parameter $(c \alpha(A_1), \dots, c\alpha(A_N))$. In addition, the distribution $\mu$ of $f$ under $\alpha$ is a discrete probability measure 
\[
	\mu = \sum_{i=1}^N \delta_{a_i} \alpha(A_i),
\] 
from which the desired MKR follows from Example~\ref{ex:Dirichlet-distribution}.

Now let $f$ be a measurable function satisfying the integrability condition~\eqref{integrability}. Take a sequence of simple functions $\{f_N\}$ such that for $u \in \R$,
\[
	f_N(u) \to f(u) \quad \text{as}\quad N \to \infty, \quad |f_N(u)| \le |f(u)|.
\]
Since $f_N$ is a simple function, we have just shown that 
\begin{equation}\label{MKR-fN}
	\Ex[(z - \bra{P, f_N})^{-c}] = \exp\left( -c \int_\R \log(z - f_N(u))d\alpha (u)\right).
\end{equation}
Recall that the condition~\eqref{integrability} implies that $\bra{P, |f|}$ is finite, almost surely \cite{Feigin-Tweedie-1989}. Thus, by the Lebesgue dominated convergence theorem, 
\[
	\bra{P, f_N} \to \bra{P, f} \quad \text{as} \quad N \to \infty, \quad \text{almost surely}.
\]
Using the dominated convergence theorem again, we deduce that for $z \in \C \setminus \R$,
\[
	 \int_\R \log(z - f_N(u))d\alpha (u) \to  \int_\R \log(z - f(u))d\alpha (u) \quad \text{as} \quad N \to \infty.
\]
Therefore, let $N \to \infty$ in equation~\eqref{MKR-fN} and use the above two limits, we get the desired relation. The proof is complete.
\end{proof}

\section{Proof of Theorem~\ref{thm:intro-general}}
\label{sect:proof}

In Example~\ref{ex:Dirichlet-distribution}, when the Dirichlet distribution is symmetric, the measure $\mu$ is the empirical distribution of $a_1, \dots, a_N$. We formulate that special case in the following example.

\begin{example}\label{ex:symmetric}
Let $w = (w_1, \dots, w_N)$ have the symmetric Dirichlet distribution with parameter $c/N$. For $a = \{a_1, \dots, a_N\} \in \R^N$, we call a random probability measure 
\[
	\nu_N = \sum_{i=1}^N w_i \delta_{a_i}
\]
the spectral measure of $a$, and $L_N = N^{-1}\sum_{i=1}^N \delta_{a_i}$ the empirical measure of $a$. Then for any function $f \colon \R \to \R$, the following Markov--Krein relation  holds, 
\begin{equation}
	\Ex[(z - \bra{\nu_N, f})^{-c}] = \exp(-c \bra{L_N, \log(z - f(\cdot))}).
\end{equation}
\end{example}

\begin{proof}[Proof of Theorem~\rm\ref{thm:intro-general}]
Let $P\sim \DP(\rho, c)$ be a Dirichlet process with base distribution $\rho$ and scaling parameter $c$. We aim to show that for any bounded continuous function $f \colon \R \to \R$, 
\[
	\bra{sp_N, f} \dto \bra{P, f} \quad \text{as} \quad N \to \infty.
\]
It is worth mentioning that the convergence in distribution ($\dto$) of random variables is equivalent to the weak convergence ($\wto$) of their distributions. Thus, it suffices to prove that any limit point of the tight sequence $\{\bra{sp_N, f}\}$ has the same distribution with $\bra{P, f}$ which is the MKT of the distribution $\mu$ of $f$ under $\rho$.

Let $\nu$ be a weak limit of the sequence $\bra{sp_N, f}$. Without loss of generality, assume that the whole sequence converges to $\nu$, that is, 
\begin{equation}\label{weak-limit}
	\bra{sp_N, f}	\dto \nu \quad \text{as} \quad N \to \infty.
\end{equation}
For each $N$, it follows from Example~\ref{ex:symmetric} that
\begin{equation}\label{MKR-finite}
	\Ex[(z - \bra{sp_N, f})^{-c}] = \Ex\Big[\exp \big \{ -c \bra{L_N, \log(z - f(\cdot))} \big \} \Big], \quad z \in \C \setminus \R.
\end{equation}
Now since $\log(z - f(x))$ is a bounded continuous function, we deduce from the assumption that  as $N \to \infty$, 
\[
	\bra{L_N, \log(z - f(\cdot))} \to \bra{\rho, \log(z - f(\cdot))} =  \int \log(z - u) d\mu(u), \quad \text{in probability}.
\]
On the one hand, the right hand side of \eqref{MKR-finite} is convergent,
\begin{align*}
	 \Ex\Big[\exp \big \{ -c \bra{L_N, \log(z - f(\cdot))} \big \} \Big] 
	 &=  \exp \left( -c \int \log(z - u) d\mu(u)\right),
\end{align*}
as a consequence of the continuous mapping theorem and the bounded convergence theorem. On the other hand, the left hand side converges to $\int_\R (z - x)^{-c} d\nu(x) $ because of the weak convergence~\eqref{weak-limit}. We conclude that 
\[
	\int_\R \frac{1}{(z - x)^c} d\nu(x) = \exp \left( -c \int \log(z - u) d\mu(u)\right), \quad z \in \C \setminus \R.
\]
In other words, the weak limit $\nu$ coincides with the MKT of $\mu$. The proof is complete.
\end{proof}

\section{Classical beta ensembles on the real line}
\label{sect:classical}

\subsection{Gaussian beta ensembles}
In this subsection, we finish the proof of Theorem~\ref{thm:main-intro}. 
As mentioned in the introduction, in the high temperature regime $\beta = 2c/N$, the empirical distribution $L_N = N^{-1}\sum_{i=1}^N \delta_{\lambda_i}$ of the eigenvalues $\lambda_1, \dots, \lambda_N$ of $H_N$ converges weakly to a limiting probability measure $\rho_c$, almost surely. 
In addition, the spectral measure $sp_N$ of $H_{N}$ is written as
\begin{equation}\label{spN}
	sp_N = \sum_{i=1}^N w_i \delta_{\lambda_i},
\end{equation}
where the weights $(w_1, \dots, w_N)$ are independent of the eigenvalues $\lambda_i$'s and have the symmetric Dirichlet distribution with parameter $c/N$. Thus, Theorem~\ref{thm:intro-general} implies that the spectral measure $sp_N$ converges in distribution to a Dirichlet process $\DP(\rho_c, c)$. To prove Theorem~\ref{thm:main-intro}, it remains to show the following.
\begin{lemma}\label{lem:nuc}
	The spectral measure $sp_N$ of $H_N$ converges in distribution  to the spectral measure $\nu_c$ of $H_c$.
\end{lemma}
\begin{proof}
We use equivalent conditions of the convergence in distribution of random probability measures stated in Theorem~\ref{thm:moment-distribution}. We first need to verify the assumption of Theorem~\ref{thm:moment-distribution} on the mean measure $\overline{sp}_N$. Although we can do it directly, we use here the fact that the mean of $sp_N$ coincides with the mean of the empirical distribution $L_N$. And thus, the assumption follows from existing results on the empirical distributions. Here the mean of $sp_N$ coincides with the mean of $L_N$ because for any bounded continuous function $f$, 
\begin{align*}
	\bra{\overline {sp}_N, f} &= \Ex[\bra{sp_N, f}] \\
	&= \Ex \left[\sum_{i=1}^N w_i f(\lambda_i)\right] \\
	&= \sum_{i=1}^N \Ex[w_i] \Ex[f(\lambda_i)] \quad (\text{$(w_i)_i$ and $(\lambda_i)_i$ are independent}) \\
	&= \sum_{i=1}^N \frac1N \Ex[f(\lambda_i)] \quad (\text{$(w_i)_i$ has the symmetric Dirichlet distribution)} \\
	&= \Ex[\bra{ L_N, f}] = \bra{\bar L_N, f}.
\end{align*}

Now let $p$ be a polynomial. Then $\bra{sp_N, p} $ can be written as 
\[
	\bra{sp_N, p} = p(H_N)(1,1)= f(a_1^{(N)}, \dots, a_k^{(N)}, b_1^{(N)}, \dots, b_k^{(N)}),
\]
for some polynomial $f$ of $2k$ variables, when $N$ is large enough. It then follows from the joint convergence of the entries and the continuous mapping theorem that 
\[
	 f(a_1^{(N)}, \dots, a_k^{(N)}, b_1^{(N)}, \dots, b_k^{(N)}) \dto  f(a_1^{(\infty)}, \dots, a_k^{(\infty)}, b_1^{(\infty)}, \dots, b_k^{(\infty)}) = \bra{\nu_c, p}.
\]
Therefore, the statement (i) in Theorem~\ref{thm:moment-distribution} holds true, implying the statement (ii), or $sp_N$ converges in distribution to $\nu_c$. The lemma is proved.
\end{proof}

\begin{remark}
	An alternative proof of Lemma~\ref{lem:nuc} using the fact that the almost sure convergence implies the convergence in distribution is as follows. There is a realization such that almost surely, as $N \to \infty$,
\[
	a_i^{(N)} \to a_i^{(\infty)}, \quad b_i^{(N)} \to b_i^{(\infty)}, 
\]
for all $i$. In this realization, similar to the above argument using the continuous mapping theorem, moments of $sp_N$ converges to the corresponding moment of $\nu_c$ almost surely as $N \to \infty$. Therefore, it follows from Proposition~\ref{prop:moments-convergence} that the spectral measure $sp_N$ converges weakly to $\nu_c$, almost surely. The convergence in distribution then follows immediately.
\end{remark}

\subsection{Beta Laguerre ensembles}
The random matrix model for beta Laguerre ensembles was introduced in \cite{DE02}.
Let $J_N^{(L)} := B_N (B_N)^\top$ be a Jacobi matrix, where $B_N$ is a bidiagonal matrix consisting of independent random variables with distributions
\[
	B_N = \begin{pmatrix}
		\chitilde_{2\alpha + \beta(N-1)}	\\
		\chitilde_{\beta (N-1)}	&\chitilde_{2\alpha + \beta(N-2)}		\\
		&\ddots	&\ddots	\\
		&&\chitilde_{\beta}	&\chitilde_{2\alpha}
	\end{pmatrix}.
\]
Here $(B_N)^\top$ denotes the transpose of $B_N$, and 
$\alpha, \beta  > 0$.
Then the joint density of the eigenvalues of $J_N^{(L)}$ is proportional to 
\begin{equation}\label{bLE}
	 \prod_{i<j}|\lambda_j - \lambda_i|^{\beta}\prod_{l = 1}^N \left(\lambda_l^{\alpha-1} e^{- \lambda_l} \right),\quad  (\lambda_i > 0).
\end{equation}

For beta Laguerre ensembles, the high temperature regime refers to the case where $\alpha$ is fixed, and $\beta = 2c/N$ as $N \to \infty$.
In this regime, the empirical distribution $L_N$ of the eigenvalues converges weakly to a limiting probability measure $\rho_{\alpha, c}$, almost surely \cite{Allez-Wishart-2013, Trinh-Trinh-2021}. The measure $\rho_{\alpha, c}$ which is the spectral measure of the following Jacobi matrix 
\[
	 \begin{pmatrix}
		\sqrt{\alpha + c}		\\
		\sqrt{c+1}	&\sqrt{\alpha + c + 1}\\
		&\sqrt{c+2}	&\sqrt{\alpha + c + 2}\\
		&&\ddots	&\ddots	
	\end{pmatrix}
	\begin{pmatrix}
		\sqrt{\alpha + c}		&\sqrt{c+1}\\
			&\sqrt{\alpha + c + 1}		&\sqrt{c+2}\\
					&&\ddots	&\ddots	
	\end{pmatrix}
\]
is called the probability measure of Model II of associated Laguerre polynomials.

For fixed $\alpha > 0$ and $c > 0$, let 
\begin{equation}
	 B_{\alpha, c} = \begin{pmatrix}
		\chitilde_{2(\alpha + c)}	\\
		\chitilde_{2c}	&\chitilde_{2(\alpha + c)}\\
		&\chitilde_{2c}	&\chitilde_{2(\alpha + c)}\\
		&&	\ddots	&\ddots	
	\end{pmatrix}
\end{equation}
be a random bidiagonal matrix consisting of two independent sequences of i.i.d.\ random variables with distribution $\chitilde_{2c}$ and $\chitilde_{2(\alpha + c)}$. Define 
\begin{equation}
J_{\alpha, c} = B_{\alpha, c}(B_{\alpha, c})^\top,
\end{equation}
which is the limit in the high temperature regime $\beta = 2c/N$ of $J_N^{(L)}$.

\begin{theorem}
Let $\alpha , c > 0$. Then the spectral measure $sp_{\alpha, c}$ of the random Jacobi matrix $J_{\alpha, c}$ is a Dirichlet process $\DP(\rho_{\alpha, c}, c)$.
\end{theorem}
We skip the proof because arguments are exactly the same as those used in the Gaussian case.

\subsection{Beta Jacobi ensembles}
We shortly mention the result in the Jacobi case. 
Let $p_1, \dots, p_N$ and $q_1, \dots, q_{N - 1}$ be independent random variables with beta distributions
\begin{align*}
	p_n &\sim \Beta((N - n) \frac \beta 2 + a, (N - n) \frac \beta 2 + b),\\
	q_n &\sim \Beta((N - n) \frac \beta 2, (N - n - 1) \frac \beta 2 + a + b),
\end{align*}
where $a, b > 0$ and $\beta > 0$.  
Let 
\begin{align*}
	s_n &= \sqrt{p_n(1 - q_{n - 1})},\quad n = 1, \dots, N, \quad (q_0 = 0),\\
	t_n &= \sqrt{q_n(1 - p_n)}, \quad n = 1, \dots, N - 1.
\end{align*}
Then the tridiagonal random matrix
\[
	J_{N}^{(J)} =  \begin{pmatrix}
		s_1	\\
		t_1	&s_2		\\
		&\ddots	&\ddots \\
		&& t_{N - 1}	& s_N
	\end{pmatrix}
	 \begin{pmatrix}
		s_1 	&t_1	\\
			&s_2		&t_2		\\
		&&\ddots	&\ddots \\
		&&& s_N
	\end{pmatrix}
\]
has the eigenvalues $(\lambda_1, \dots, \lambda_N)$ distributed as the beta Jacobi ensemble \cite{Killip-Nenciu-2004}, that is, the joint density of the eigenvalues is proportional to 
\[
	 \prod_{i < j} |\lambda_j - \lambda_i|^\beta \prod_{l = 1}^N \lambda_l^{a-1} (1 - \lambda_l)^{b-1}, \quad \lambda_i \in [0,1].
\]

In the high temperature regime where $a, b > 0$ are fixed and $\beta = 2c / N$, the empirical distribution converges weakly to a limiting measure $\rho_{a,b,c}$, almost surely, which is the probability measure of Modell III of associated Jacobi polynomials \cite{Trinh-Trinh-Jacobi}.

The limit of the Jacobi matrix $J_{N}^{(J)}$ is formed from two sequences of i.i.d. random variables 
\begin{align*}
	p_n^{(\infty)} &\sim \Beta(c + a, c + b),\\
	q_n^{(\infty)} &\sim \Beta(c, c+ a + b),
\end{align*}
in the same way as $J_N^{(J)}$, that is, 
\begin{align*}
	s_n^{(\infty)} &= \sqrt{p_n^{(\infty)}(1 - q_{n - 1}^{(\infty)})},\quad n = 1, 2, \dots, \quad (q_0^{(\infty)} = 0),\\
	t_n^{(\infty)} &= \sqrt{q_n^{(\infty)}(1 - p_n^{(\infty)})}, \quad n = 1,2, \dots,
\end{align*}
and 
\begin{equation}
	J_{a,b,c} =  \begin{pmatrix}
		s_1^{(\infty)}	\\
		t_1^{(\infty)}	&s_2^{(\infty)}		\\
		&\ddots	&\ddots \\
	\end{pmatrix}
	 \begin{pmatrix}
		s_1^{(\infty)} 	&t_1^{(\infty)}	\\
			&s_2^{(\infty)}		&t_2^{(\infty)}		\\
		&&\ddots	&\ddots \\
	\end{pmatrix}.
\end{equation}
Let $\nu_{a,b,c}$ be the spectral measure of $J_{a, b, c}$. Then we have the following result.
\begin{theorem}
The spectral measure $\nu_{a,b,c}$ is the Dirichlet process $\DP(\rho_{a, b, c}, c)$.
\end{theorem}

\section{Concluding remarks on beta ensembles on the real line}
Consider the following beta ensembles
\begin{equation}\label{VbE}
	(\lambda_1, \lambda_2, \dots, \lambda_N) \propto \frac{1}{Z_{\beta, N}} \left( \prod_{i<j}|\lambda_j - \lambda_i|^\beta \right) \times e^{-\sum_{i = 1}^N V(\lambda_i)},
\end{equation}
where $Z_{N, \beta}$ is the normalizing constant. 	Assume that the potential $V$ is continuous and that 
	\[
		\lim_{x \to \pm \infty} \frac{V(x)}{\log (1+x^2)} = \infty.
	\]
Then in the regime where $\beta N \to 2c \in (0, \infty)$, 
the sequence of empirical distributions $\{L_N\}$ converges weakly to a probability measure $\rho_c$, almost surely \cite{Nakano-Trinh-2020}. The measure $\rho_c$ is the unique minimizer of the functional $\E_c$ defined for absolutely continuous probability measure $\mu(dx) = \rho(x) dx$,
\[
	\E_c(\mu) = \int \log(\rho(x)) \rho(x) dx + \int V(x) \rho(x) dx  - c \int \log|x - y| \rho(x) \rho(y) dx dy.
\]
Moreover, $\rho_c$ has a continuous density $\rho_c(x)$ satisfying the relation 
\[
\rho_c(x) = \frac{1}{Z_{c} }e^{-V(x) + 2 c  \int \log|x - y| \rho_c(y) dy}, \quad \text{for all $x \in \R$,}
\]
where $Z_c$ is a constant.

A random tridiagonal model for those beta ensembles can also be constructed.
For each $N$, consider the Jacobi matrix 
\[
	J_N = \begin{pmatrix}
		a_1^{(N)}	&b_1^{(N)}\\
		b_1^{(N)}	&a_2^{(N)}		&b_2^{(N)}\\
		&\ddots		&\ddots	&\ddots\\
		&&b_{N-1}^{(N)}	&a_N^{(N)}
	\end{pmatrix},
\]
where $\{a_i\}_{i=1}^N$ and $\{b_j\}_{j=1}^{N-1}$ have the joint density proportional to 
\begin{equation}\label{joint-a-b}
 e^{- \tr(V(J_N))}\prod_{j=1}^{N-1} b_j^{\beta(N-j) - 1} , \quad a_i \in \R, b_j > 0.
\end{equation}
Here for simplicity, we have removed the superscript ${}^{(N)}$ in the above formula.
Then the eigenvalues of $J_N$ follow the beta ensemble~\eqref{VbE}. Moreover, the spectral measure of $J_N$ has the form,
\[
	sp_N = \sum_{i=1}^N w_i \delta_{\lambda_i},
\]
where $w = (w_1, \dots, w_N)$ has the symmetric Dirichlet distribution with parameter $\frac \beta 2$ and is independent of $\{\lambda_i\}_{i=1}^N$.

As a consequence of Theorem~\ref{thm:intro-general}, the spectral measure $sp_N$ converges weakly to a Dirichlet process $P \sim \DP(\rho_c, c)$ with base distribution $\rho_c$ and scaling parameter $c$.
Construct the Jacobi matrix of the (random) probability measure $P$
\[
	\cJ_c = \begin{pmatrix}
		a_1^{(\infty)}	&b_1^{(\infty)}\\
		b_1^{(\infty)}	&a_2^{(\infty)}	&b_2^{(\infty)}\\
		&\ddots		&\ddots	&\ddots
	\end{pmatrix}.
\]
The convergence of the spectral measures $sp_N$ implies that the random matrix $J_N$ converges in distribution to $\cJ_c$ as stated in the following theorem.

\begin{theorem}
	For any $n$, 
	\[
		(a_1^{(N)}, b_1^{(N)},  \dots, a_n^{(N)}, b_n^{(N)})	\dto (a_1^{(\infty)}, b_1^{(\infty)},  \dots, a_n^{(\infty)}, b_n^{(\infty)})\quad \text{as}\quad N \to \infty.
	\]
\end{theorem}
\begin{remark}
Note that in the three classical cases, the limiting Jacobi matrix $\cJ_c$ naturally appears as the limit of the sequence of random Jacobi matrix $J_N$. However,  $\cJ_c$ here has been defined from the Dirichlet process $P$. It would be interesting to explicitly describe the joint distribution of entries of $\cJ_c$.
\end{remark}
\begin{proof}

Theorem~\ref{thm:moment-distribution} implies that for any $k$, 
\[
	(\bra{sp_N, x}, \dots, \bra{sp_N, x^k}) \dto (\bra{P, x}, \dots, \bra{P, x^k}).
\]
Consequently, it holds that $a_1^{(N)} = \bra{sp_N, x} \dto a_1^{(\infty)} = \bra{P,x}$. Next, since 
\[
	\bra{sp_N, x^2} = (a_1^{(N)})^2 +  (b_1^{(N)})^2,
\]
it implies that 
\[
	(a_1^{(N)}, b_1^{(N)}) \dto (a_1^{(\infty)}, b_1^{(\infty)}),
\]
by the continuous mapping theorem. In general, note that for $k = 2n+1$,
\[
	\bra{sp_N, x^{2n+1}} = a_{n+1}^{(N)}\prod_{j=1}^{n} (b_j^{(N)})^2  + \text{polynomial\,}\left(\{a_i^{(N)}, b_i^{(N)}\}_{i=1}^n\right),
\]
and that for $k = 2(n+1)$, 
\[
	\bra{sp_N, x^{2n+2}} = (b_{n+1}^{(N)})^2\prod_{j=1}^{n} (b_j^{(N)})^2  + \text{polynomial\,}\left(\{a_i^{(N)}, b_i^{(N)}\}_{i=1}^n\right).
\]
From which, the desired joint convergence follows from the continuous mapping theorem by taking into account of $b_j^{(\infty)} > 0$, almost surely. The proof is complete.
\end{proof}

\appendix
\section{Convergence of random probability measures on the real line}\label{sect:rpm}

We first introduce equivalent definitions of convergence of random probability measures in terms of test functions. Refer to \cite[Section 4]{Kallenberg} for the part on convergence in distribution.
\begin{definition}\label{def:rm}
\begin{itemize}
\item[(i)]
	Let $\{\xi_N\}_{N = 1}^\infty$ and $\xi$ be random probability measures defined on the same probability space $(\Omega, \F, \Prob)$. The sequence $\{\xi_N\}$ is said to converge weakly to $\xi$, almost surely  if for any bounded continuous function $f$, 
	\[
		\bra{\xi_N, f} \to \bra{\xi, f} \quad \text{as}\quad N \to \infty, \quad \text{almost surely.}
	\]

\item[(ii)] Let $\{\xi_N\}_{N = 1}^\infty$ and $\xi$ be random probability measures defined on the same probability space $(\Omega, \F, \Prob)$. The sequence $\{\xi_N\}$ is said to converge weakly to $\xi$, in probability if for any bounded continuous function $f$, 
	\[
		\bra{\xi_N, f} \to \bra{\xi, f} \quad \text{in probability as}\quad N \to \infty.
		\]
When $\xi$ is non-random, the condition that all random probability measures are defined on the same probability space is not necessary.

\item[(iii)] 
	Let $\{\xi_N\}_{N =1}^\infty$ and $\xi$ be random probability measures which may be defined on different probability spaces. The sequence $\{\xi_N\}$ is said to converge in distribution to $\xi$ if for any bounded continuous function $f$, 
	\[
		\bra{\xi_N, f} \dto \bra{\xi, f} \quad \text{as}\quad N \to \infty.
	\]

\end{itemize}
\end{definition}

It is well known that the convergence of moments implies the weak convergence of probability measures, provided that the limiting measure is determined by moments. We have an analogous result for random probability measures whose proof can be found in \cite{Trinh-ojm-2018}, for example.
\begin{proposition}\label{prop:moments-convergence}
Let $\{\xi_N\}_{N = 1}^\infty$ and $\xi$ be random probability measures defined on the same probability space $(\Omega, \F, \Prob)$.
Assume that the random probability measure $\xi$ is determined by moments, almost surely. Then the condition 
\[
	\bra{ \xi_N, x^k} \to \bra{ \xi, x^k} \text{ almost surely, for } k = 0,1, \dots,
\]
implies that $\{\xi_N\}$ converges weakly to $\xi$, almost surely. Here we have assumed that every moment is finite, almost surely. Moreover,  if $f$ is a continuous function of polynomial growth, that is, there is a polynomial $p$ such that $|f(x)|\le p(x)$ for all $x\in \R$, then it holds that 
\[
	\bra{ \xi_N, f } \to \bra{ \xi, f}	\quad \text{as} \quad N\to \infty, \quad \text{almost surely.}
\]
The statement still holds true, if we replace the almost sure convergence by the convergence in probability.
\end{proposition}

We are now in a position to deal with the convergence of moments in the convergence in distribution.
Let $\xi$ be a random probability measure. The mean of $\xi$, denoted by $\bar \xi$, is  a probability measure defined by 
\[
	\bar \xi(B) = \Ex[\xi(B)], \quad B \in \cB(\R).
\]
It turns out that 
\[
	\bra{\bar \xi, f} = \Ex[\bra{\xi, f}],
\]
for any bounded measurable function $f \colon \R \to \R$, and any non-negative measurable function $f\colon \R \to [0, \infty)$.

Assume that the mean probability measure $\bar \xi$ has all finite moments. Then all moments of $\xi$ are finite, almost surely, and it holds that
\[
	\bra{\bar \xi, x^n} = \Ex[\bra{\xi, x^n}], \quad n = 0,1,2,\dots, 
\]
or 
\[
	\bra{\bar \xi, f} = \Ex[\bra{\xi, f}],
\]
for continuous functions $f$ of polynomial growth.

\begin{theorem}\label{thm:moment-distribution}
Let $\{\xi_N\}_{N = 1}^\infty$ and $\xi$ be random probability measures such that their mean probability measures have all finite moments. Assume that for any $k = 0,1,2,\dots,$
\[
	\bra{\bar \xi_N, x^k} \to \bra{\bar \xi, x^k} \quad \text{as} \quad N \to \infty,
\]
and that  $\bar \xi$ is determined by moments. Then the following are equivalent.
\begin{itemize}
\item [\rm (i)]
For any polynomial $p$, 
\begin{equation}\label{moment-convergence}
	\bra{\xi_N, p} \dto \bra{\xi, p} \quad \text{as} \quad N \to \infty.
\end{equation}

\item[\rm(ii)] $\xi_N$ converges in distribution to $\xi$ as $N \to \infty$.

\end{itemize}
\end{theorem}

\begin{remark}
The statement (i) is equivalent to the condition that for any $k$, 
\[
	(\bra{\xi_N, x}, \dots, \bra{\xi_N, x^k}) \dto (\bra{\xi, x}, \dots, \bra{\xi, x^k})
\]
as the joint convergence in distribution of random variables.
\end{remark}

To prove this theorem, we will use the following approximation argument.
\begin{lemma}[{\cite[Theorem 25.5]{Billingsley}}]\label{lem:triangle}
Let $\{Y_N\}_N$ and $\{X_{N,k}\}_{N, k}$ be real-valued random variables. Assume that
\begin{itemize}
	\item[\rm(a)] 
		$
			X_{N,k} \dto X_k  \text{ as }N \to \infty;
		$
	\item[\rm(b)]
		$
			X_k \dto X \text{ as }  k \to \infty;
		$
	\item[\rm(c)] for any $\varepsilon > 0$,
		$
			\lim_{k \to \infty} \limsup_{N \to \infty} \Prob(|X_{N,k} - Y_N| \ge \varepsilon) =0.
		$
\end{itemize}
Then $Y_N \dto X$ as $N \to \infty$.
\end{lemma}

\begin{proof}[Proof of Theorem~{\rm\ref{thm:moment-distribution}}]
(i) $\Rightarrow$ (ii).
From the assumption, each moment of $\bar \xi_N$ converges to the corresponding moment of $\bar \xi$, and the measure $\bar \xi$ is determined by moments. This implies that $\bar \xi_N$ converges weakly to $\bar \xi$ as $N \to \infty$. In other words, for any bounded continuous function $f$,
\begin{equation}\label{poly-growth}
	\bra{\bar \xi_N, f} \to \bra{\bar \xi, f} \quad \text{as} \quad N \to \infty.
\end{equation}
Note that the above also holds for continuous functions of polynomial growth.

	Let $f$ be a bounded continuous function. If suffices to show that 
\begin{equation}\label{f-bc}
	\bra{\xi_N, f} \dto \bra{\xi, f} \quad \text{as} \quad N \to \infty.
\end{equation}
We are going to use the approximation argument in Lemma~\ref{lem:triangle}. 
	Take a sequence of polynomials $\{p_k\}_{k=1}^\infty$ approximating $f$ in the sense that 
\[
	\bra{\bar \xi, (f - p_k)^2} \to 0 \quad \text{as} \quad k \to \infty.
\]
Such sequence exists because when the measure $\bar \xi$ is determined by moments, polynomials are dense in the $L^2$-space $L^2(\bar \xi)$. Let $Y_N = \bra{\xi_N, f}, X_{N,k} = \bra{\xi_N, p_k}, X_k = \bra{\xi, p_k}$ and $X = \bra{\xi, f}$. The desired convergence~\eqref{f-bc} follows once we finish checking the three conditions in Lemma~\ref{lem:triangle}.

First, the condition~(a) is exactly (i),
\begin{equation}\label{N-to-infinity}
	\bra{\xi_N, p_k} \dto \bra{\xi, p_k} \quad \text{as}\quad N \to \infty.
\end{equation}
Next, we observe that
\[
	\Ex[(\bra{\xi, p_k} - \bra{\xi, f})^2] = \Ex[\bra{\xi, (p_k - f)}^2] \le \Ex[\bra{\xi, (p_k - f)^2}] = \bra{\bar \xi, (p_k - f)^2}, 
\]
which clearly goes to zero as $k \to \infty$. The condition (b) follows because the $L^2$ convergence implies the convergence in distribution. Now, by using the same estimate, we deduce that
\[
	\limsup_{N \to \infty}\Ex[(\bra{\xi_N, p_k} - \bra{\xi_N, f})^2] \le \limsup_{N \to \infty}  \bra{\bar\xi_N, (p_k - f)^2} = \bra{\bar \xi, (p_k - f)^2}.
\]
Here we have used equation~\eqref{poly-growth} in the last equality, because the function $(p_k - f)^2$ is continuous of polynomial growth. That estimate, together with Markov's inequality, implies that for any $\varepsilon > 0$,
\begin{equation}\label{epsilon-approximation}
	\lim_{k \to \infty}\limsup_{N \to \infty}\Prob \left(|\bra{\xi_N, p_k} - \bra{\xi_N, f}| \ge \varepsilon \right) = 0.
\end{equation}
The condition (c) has been shown, which completes the proof of [(i) $\Rightarrow$ (ii)].

(ii) $\Rightarrow$ (i). Let $p$ be a polynomial. For $k > 0$, consider the truncation of $p$,
	\[
		t_k(x) = \begin{cases}
			p(x), &\text{if }|p(x)| \le k,\\
			k,   &\text{if }p(x) > k,\\
			-k,  &\text{if }p(x) <- k.
			\end{cases}
	\]
We use Lemma~\ref{lem:triangle} again to show that 
\[
	Y_N := \bra{\xi_N, p} \dto X:= \bra{\xi, p}\quad \text{as} \quad N \to \infty,
\]
via the sequence $X_{N, k} := \bra{\xi_N, t_k}$ and the sequence $X_k := \bra{\xi, t_k}$. For that, we need to check three conditions (a), (b) and (c) in the lemma.

The condition (a), 
\[
	\bra{\xi_N, t_k} \dto \bra{\xi, t_k},
\]
 is clear since  $t_k$ is a bounded continuous function.
For (b), note that 
\[
	\bra{\xi, p} < \infty, \quad \text{almost surely},
\]
by the assumption. Then by the Lebesgue dominated convergence theorem, 
\[
	\bra{\xi, t_k} \to \bra{\xi, p} \quad \text{as}\quad k \to \infty, \quad \text{almost surely,}
\]
implying (b). It remains to show (c). For $\varepsilon > 0$, we first estimate that
\begin{align*}
	\Prob(|X_{N, k} - Y_N| \ge \varepsilon) &= \Prob(|\bra{\xi_N, p - t_k}| \ge \varepsilon) \\
	&\le \frac{1}{\varepsilon} \Ex[\bra{\xi_N, |p-t_k|}] = \frac1{\varepsilon} \bra{\bar \xi_N, |p-t_k|}.
\end{align*}
Clearly, the function $ |p-t_k|$ is continuous of polynomial growth. Thus,
take the limit as $N \to \infty$ in the above equation and use~\eqref{poly-growth}, we get that 
\[
	\limsup_{N\to \infty}\Prob(|X_{N, k} - Y_N| \ge \varepsilon) \le \frac1{\varepsilon} \bra{\bar \xi, |p-t_k|},
\]
which tends to zero as $k \to \infty$. The condition (c) is satisfied, which completes the proof of Theorem~\ref{thm:moment-distribution}.
\end{proof}

\subsection*{Acknowledgments.}
This research is supported by the VNU University of Science, grant number TN.25.17 (H.D.T) and by
JSPS KAKENHI Grant number JP24K06766 (K.D.T.).

%

\end{document}